\documentclass[conference]{IEEEtran}
\usepackage[utf8]{inputenc}
\usepackage{ amssymb }
\usepackage{amsthm}
\usepackage{siunitx}
\usepackage{tikz}
\usepackage{comment}
\usetikzlibrary{shapes, arrows, arrows.meta, calc, positioning, fit,backgrounds}
\usepackage{pgfplots}
\pgfplotsset{compat=newest}
\usetikzlibrary{patterns}
\usepackage{url}

\title{XOX Fabric: A hybrid approach to blockchain transaction execution}
\author{\IEEEauthorblockN{Christian Gorenflo}
\IEEEauthorblockA{
University of Waterloo\\
Waterloo, Canada\\
Email: cgorenflo@uwaterloo.ca}
\and
\IEEEauthorblockN{Lukasz Golab}
\IEEEauthorblockA{
University of Waterloo\\
Waterloo, Canada\\
Email: lgolab@uwaterloo.ca}
\and
\IEEEauthorblockN{Srinivasan Keshav}
\IEEEauthorblockA{
University of Cambridge\\
Cambridge, UK\\
Email: sk818@cam.ac.uk}}

\date{June 2019}
\newtheorem*{theorem}{Hot Key Theorem}

\begin{document}
\maketitle
\begin{abstract}
    Performance and scalability are major concerns for blockchains: permissionless systems are typically limited by slow \emph{proof~of~X} consensus algorithms and sequential post-order transaction execution on every node of the network. By introducing a small amount of trust in their participants, permissioned blockchain systems such as Hyperledger Fabric can benefit from more efficient consensus algorithms and make use of parallel pre-order execution on a subset of network nodes. Fabric, in particular, has been shown to handle tens of thousands of transactions per second. However, this performance is only achievable for contention-free transaction workloads. If many transactions compete for a small set of hot keys in the world state, the effective throughput drops drastically. We therefore propose XOX: a novel two-pronged transaction execution approach that both minimizes invalid transactions in the Fabric blockchain and maximizes concurrent execution. Our approach additionally prevents unintentional denial of service attacks by clients re-submitting conflicting  transactions. Even under \textit{fully} contentious workloads, XOX can handle more than 3000 transactions per second, \textit{all} of which would be discarded by regular Fabric.
\end{abstract}

\section{Introduction}
Blockchain systems have substantially evolved from their beginnings as tamper-evident append-only logs. With the addition of \emph{smart contracts}, complex computations based on the blockchain's state become possible. In fact, several  permissionless and permissioned systems such as Ethereum~\cite{Founder2017} and Hyperledger Fabric~\cite{Androulaki2018a} allow Turing-complete computations. However, uncoordinated execution of smart contracts in a decentralized network can result in inconsistent blockchains, a fatal flaw. Fundamentally,  blockchain systems have two options to resolve such conflicts. They can either coordinate, i.e., execute contracts after establishing consensus on a linear ordering, or they can deterministically resolve inconsistencies after parallel execution.

Most existing blockchain systems implement smart contract execution after ordering transactions, giving this pattern the name \emph{order-execute}~(OX). In these systems, smart contract execution happens sequentially. 
This allows each execution to act on the result of the previous execution, but restricts the computation to a single thread, limiting performance. Blockchains using this pattern must additionally guarantee that the smart contract execution reaches the same result on every node in the network that replicates the chain, typically by requiring smart contracts to be written in a domain-specific deterministic programming language. This restricts programmability. Moreover, this makes the use of external data sources, so-called \emph{oracles}, difficult, because they cannot be directly controlled and may deliver different data to different nodes in the network.

Other blockchain systems,  most notably Hyperledger Fabric, use an \emph{execute-order}~(XO) pattern. Here, smart contracts referred to by transactions are executed in parallel in a container before the ordering phase. Subsequently, only the results of these computations are ordered and put into the blockchain. Parallelized smart contract execution allows, among other benefits, a nominal transaction throughput orders of magnitude higher than that of other blockchains~\cite{Dinh2017d}. However, a model that executes each transaction in parallel is inherently unable to detect transaction conflicts\footnote{Two transactions are said to conflict if either one reads or writes to a key that is written to by the other.} during execution.
%as illustrated by the following example. Take a smart contract that transfers digital coins from one account to another. Assume one transaction tries to add 40 coins to an account currently holding 100 coins,  while in close succession another transaction subtracts 20 coins from the same account. One will calculate 140 coins as the account balance and the other 80 coins, because neither has knowledge of the other transaction. Using an XO~model, Fabric cannot re-evaluate the results after they are ordered; it can only choose to accept either 140 or 80 as the final result and discard the other. To do this correctly, it has to filter out invalid transactions in a sequential validation pass after the order is known. This bottleneck decreases the effective transaction throughput to a fraction of the nominal throughput if the percentage of conflicting transactions is large.

Prior work on contentious workloads in Fabric focuses on detecting conflicting transactions during ordering and aborting them early.
%early transaction abort, that is, the detection of abortion of transactions during ordering when a conflict is detected. 
However, this tightly couples architecturally distinct parts of the Fabric network, breaking its modular structure. Furthermore, early abort only treats a symptom and not the cause in that it only filters out conflicting  transactions instead of preventing their execution in the first place. This approach will not help if many transactions try to modify a small number of `hot' keys. For example, suppose the system supports a throughput of 1000 transactions per second. Additionally, suppose 20 transactions try to access the same key in each block of 100 transactions. Then, only one of the 20 transactions will be valid and the rest must be aborted early. Subsequently, all 19 aborted clients will attempt to re-execute their transactions, adding to the 20 new conflicting transactions in the next block. This leads to 38 aborted transactions in the next round, and so on. Clearly, with \textit{cumulative re-execution}, the number of aborted transactions grows linearly until it surpasses the throughput of the system. Thus, if clients re-execute aborted transactions, their default behaviour, this effectively becomes an unintentional denial of service attack on the blockchain! 

This inherent problem with the XO pattern greatly reduces the performance of uncoordinated marketplaces or auctions. For example, conflicting transactions cannot be avoided in use cases such as payroll, where an employer transfers credits to a large number of employees periodically, or energy trading, where a small number of producers offer fungible units of energy to a large group of consumers.

We therefore propose XOX Fabric that essentially adds a second deterministic re-execution phase to Fabric. This phase executes `patch-up code' that must be added to a smart contract. We show that, in many cases, this eliminates the need to re-submit and re-execute conflicting  transactions. Our approach can deal with highly skewed contentious workloads with a handful of hot keys, while still retaining the decoupled, modular structure of Fabric. Our contributions are as follows:
\begin{itemize}
    \item \emph{Hybrid execution model:} Our \emph{execute-order-execute}~(XOX) model allows us to choose an optimal trade-off between concurrent high-performance execution and consistent linear execution, preventing the cumulative re-execution problem.
    \item \emph{Compatibility with external oracles:} To allow the use of external oracles in the deterministic second execution phase, we gather and save oracle inputs in the pre-order execution step. 
    \item \emph{Concurrent transaction processing:} By computing a DAG of transaction dependencies in the post-order phase, Fabric peers can maximize parallel transaction processing.
Specifically, they not only parallelize transaction validation and commitment, making full use of modern multi-core CPUs, but also re-execute transactions in parallel as long as these transactions are not dependent on each other. This alleviates the execution bottleneck of OX blockchains.
\end{itemize}
We achieve these contributions while being fully legacy-compatible and without affecting Fabric's modularity. In fact, XOX can replace an existing Fabric network setup by changing only the Fabric binaries\footnote{Source code available at \url{https://github.com/cgorenflo/fabric/tree/xox-1.4}}.

\section{Background}
\subsection{State machine replication and invalid state transitions}

We can model blockchain systems as state machine replication mechanisms~\cite{Schneider1990}. Each node in the network stores a replica of a state machine, with the genesis block as its \texttt{START} state. Smart contracts then become state transition functions. They take client requests (transactions) as input and compute a state transition which can be subsequently committed to the world state. This world state is either implicitly created by the data on the blockchain or explicitly tracked in a data store, most commonly a key-value store. Because of blockchain's inherently decentralized nature, keeping the world state consistent on all nodes is not trivial. A node's stale state, a client's incomplete information, parallel smart contract execution, or malicious behaviour can  produce conflicting state transitions. Therefore, a blockchain's execution model must prevent such transactions from modifying the world state. There are two possibilities to accomplish this, the OX and XO approaches, that we have already outlines. In the next two subsections, we explore some
subtleties of each approach.
%and we will describe both.
    
\subsection{The OX model}

The \emph{order-execute~(OX)} approach guarantees consensus on linearization of transactions in a block. However, it requires certain restrictions on the execution engine to  guarantee that each node makes identical  state transitions. First, the output of the execution engine must be deterministic. This requires the use of a deterministic contract language, such as Ethereum's Solidity, which must be learned by the application developer community. It also means that external oracles cannot easily be incorporated because different nodes in the network may receive different information from the oracle. Second, depending on the complexity of smart contracts, there needs to be a mechanism to deal with the \emph{halting problem}, i.e., the inherent \textit{a priori }unknowability of contract execution duration. 
A common solution to this problem is the inclusion of an execution fee like Ethereum's \textit{gas}, which aborts long-running contracts.

\subsection{The XO model}
The \emph{execute-order~(XO)} model approach allows transactions to be executed in arbitrary order: the resulting state transitions are then ordered and aggregated into blocks. This allows transactions to be executed in parallel, increasing throughput. However, the world state computed at the time of state transition commitment is known to execution engines only after some delay, and all transactions are inevitably executed on a stale view of the world state. This makes it possible for transactions to result in invalid state transitions even though they were executed successfully before ordering. It necessitates a validation step after ordering so transitions can be invalidated deterministically based on detected conflicts. Consequently, for a transaction workload with a set of frequently updated keys, the effective throughput of an XO system can be significantly lower than the nominal throughput (we formalize this as the \textit{Hot Key Theorem }below).
%The most prominent proponent of the XO model, 

\subsection{Hyperledger Fabric}
Hyperledger Fabric has been described in detail by Androulaki \emph{et al}~\cite{Androulaki2018a}. Below, we describe those parts of the Fabric architecture that are relevant to this work.

A Fabric network consists of \emph{peer} nodes replicating the blockchain and the world state, and a set of nodes called the \emph{ordering service} whose purpose is to order transactions into blocks. The world state is a key-value view of the state created by executing transactions. The nodes can belong to different organizations collaborating on the same blockchain. Because of the strict separation of concerns, Fabric's blockchain model is independent of the consensus algorithm in use. In fact, release 1.4.1 supports three plug-in algorithms, \emph{solo, Kafka} and \emph{Raft}, out of the box. As we will show in section~\ref{sec:xox} we preserve Fabric's modularity completely.

Apart from replication and ordering, Fabric needs a way to execute its equivalent of smart contracts, called \emph{chaincode}. \emph{Endorsers}, a subset of peers, fill this role. Each transaction proposal received by an endorser is \emph{simulated} in isolation. A successful simulation of arbitrarily complex chaincode results in a read and write set (RW set) of \emph{\{key, value, version\}} tuples. They act as instructions for transitions of the world state. The endorser then appends the RW set to the initial proposal, signs the response, sends it back to the requesting client, and discards the world state effect of the simulated transaction before executing the next one.

To combat non-determinism and malicious behaviour during chaincode execution, \emph{endorsement policies} can be set up. For example, a client may be required to collect identical responses from three endorsers across two different organizations before sending the transaction to the ordering service. 

After transactions have been ordered into blocks, they are disseminated to all peers in the network. These peers first independently perform a syntactic validation of the blocks and verify the endorsement policies. Lastly, they sequentially compare each transaction's RW set to the current view of the world state. If the version number of any key in the set disagrees with the world state, the transaction is discarded. Thus, any RW set overlap across transactions in the same block leads to an invalidation of all but the first conflicting transaction. As a consequence of this execution model, Fabric's blockchain also contains invalid transactions, which every peer independently flags as such during validation and ignores during commitment to world state. In the worst case, all transactions in a block might be invalid. This can drastically reduce the effective transaction throughput of the system.

\section{Related Work}
Performance is an important issue for blockchain systems since they are still slower than traditional database systems~\cite{Chen2018,Dinh2017d}. While most research focuses on consensus algorithms, less work has been done to optimize other aspects of the transaction flow, especially transaction execution.

We base this work on \emph{FastFabric}, our previous optimization of Hyperledger Fabric~\cite{Gorenflo2020}. We introduced efficient data structures, caching, and increased parallelization in the transaction validation pipeline to increase Fabric's throughput for \emph{conflict-free} transaction workloads by a factor six to seven. 
%The numbers we presented resulted from a conflict-free transaction workload. Now, we extend our findings to handle arbitrary contentious workloads.
In this paper, we address the issue of conflicting transactions.

To the best of our knowledge, a document from the Fabric  community~\cite{Sorniotti} is the first to propose a secondary post-order execution step for Fabric. However, the allowed commands were restricted
to addition, subtraction, and checking if a number is within a certain interval. Furthermore, this secondary execution step is always triggered regardless of the workload, and is not parallelized.
%circumstance and no mind is paid to parallel execution. 
This diminishes the value of retaining the first pre-order execution step and introduces the same bottleneck that OX models have to deal with.

Nasirifard \emph{et al}~\cite{Nasirifard2019} take this idea one step further. By introducing conflict-free replicated data types (CRDTs), which allow conflicting transactions to be merged during the validation step, they follow a similar path to our work. However, their solution has several limitations. It can only process transactions sequentially, one block at a time. When conflicts are discovered, they use the inherent functionality of CRDTs to resolve them. While this enables efficient computation, it also restricts the kind of execution that can be done. For example, it is not possible to check a condition like negative account balance before merging two transaction results.

Amiri \emph{et al}~\cite{Amiri2019} introduce ParBlockchain using a similar architecture to Fabric's but with an OX model. Here, the ordering service also generates a dependency graph of the transactions in a block. Subsequently, transactions in the new block are distributed to nodes in the network to be executed, taking the dependencies into account. Only a subset of nodes executes any given transaction and shares the result with the rest of the network. This approach has two drawbacks. First, the ordering service must determine the transaction dependencies before they are executed.  This requires the orderers to have complete knowledge of all installed smart contracts, and, as a result, restricts the complexity of allowed contracts. Even if a single conditional statement relies on a state value, for example \emph{Read the value of key $k$, where $k$ is the value to be read from key $k'$}, reasoning about the result becomes impossible. Second, depending on the workload, all nodes may have to communicate the current world state after every transaction execution to resolve execution deadlocks. This leads to a significant networking overhead. 

CAPER \cite{Amiri2019b} extends ParBlockchain by modelling the blockchain as a directed acyclic graph (DAG). This enables sharding of transactions. Each shard maintains an internal chain of transaction that is intertwined with a global cross-shard chain. Both internal chains and the global chain are totally ordered. This approach works well for scenarios with tightly siloed data pockets that can easily be sharded and where cross-shard transactions are rare. In this case, internal transactions of different shards can be executed in parallel. However, if the workload does not have clear boundaries to separate the shards, then most transactions will use the global chain, negating the benefit of CAPER.

Sharma \emph{et al}~\cite{Sharma2019} approach blockchains from a database point of view and incorporate concepts such as early abort and transaction reordering into Hyperledger Fabric. However, they do not follow its modular design and closely couple the different building blocks. For both early abort and transaction reordering, the ordering service needs to understand  transaction contents to unpack and analyze RW sets. Furthermore, transaction reordering only works in pathological cases. Whenever a key appears both in the read and write set, which is the case for any application that transfers any kind of asset, reordering will not eliminate RW set conflicts. While early transaction abort might increase overall throughput slightly, it cannot solve the problem of hot keys and only skews the transaction workload away from those keys.

Zhang \emph{et al}~\cite{Zhang2019} present a solution for a client-side early abort mechanism for Fabric. They introduce a transaction cache on the client that analyzes endorsed transactions to detect RW set conflicts and only sends conflict-free transactions to the ordering service. Transactions that have dependencies are held in the cache until the conflict is resolved and then they are sent back to the endorsers for re-execution. This approach prevents invalid transactions from a single client, but cannot deal with conflicts between multiple clients. Moreover, it cannot deal with hot key workloads.

Lastly, Escobar \emph{et al}~\cite{Escobar2019} investigate parallel state machine replication. They focus on efficient data structures to keep track of parallelizable, i.e., independent state transitions. While this might be interesting to incorporate into Fabric in the future, we show in section~\ref{sec:exp} that the overhead of our relatively simple implementation of a dependency tracker is negligible compared to the transaction execution overhead.

\section{The Hot Key Theorem}
We now state and prove a theorem that limits the performance of any XO system.
\begin{theorem}

	Let $\overline{l}$ be the average time between a transaction's execution and its state transition commitment. Then the average effective throughput for all transactions operating on the same key is at most $\frac{1}{\overline{l}}$.
\end{theorem}
\begin{proof}
	The proof is by induction. Let $i$ denote the number of changes to an arbitrary but fixed key $k$.
	
	$i=0$ (just before the first change):
	
	For $k$ to exist, there must be exactly one transaction $tx_0$ which takes time $l_0$ from execution to commitment and creates $k$ with version~$v_1$ at time $t_1$.
	
	$i \rightarrow i+1$ (just before the $i+1^{th}$ change):
	
	Let $k$'s current version be $v_i$ at time $t_i$. Let $tx_i$ be the first transaction in a block which updates $k$ to a new version $v_{i+1}$. The version of $k$ during $tx_i$'s execution must have been $v_i$, otherwise Fabric would invalidate $tx_i$ and prevent commitment. Let $tx_i$ be committed at time $t_{i+1}$ and $l_i$ be the time between $tx_i$'s execution and commitment.  Therefore,  
	$$t_i \leq t_{i+1} - l_i.$$
	Likewise, no transaction $tx'_i$ which is ordered after $tx_i$ can commit an update $v_i \rightarrow v'_{i+1}$ because $tx_i$ already changed the state and $tx'_i$ would therefore be invalid. Consequently, $tx_i$ must be the only transaction able to update $k$ from $v_i$ to a newer version.
	
	This means, $N$ updates to $k$ take $t_{N}$ time with
	$$t_{N} \geq \sum_{i=0}^{N-1} l_i.$$
	
	A lower bound on the average update time is then given by
	$$\frac{1}{N}t_{N}\geq\sum_{i=0}^{N-1}\frac{1}{N} l_i= \overline{l},$$
	so we get $\frac{1}{\overline{l}}$ as an upper bound on throughput being the inverse of the update latency.
	\end{proof}

This theorem has a crucial consequence. For example, FastFabric can achieve a nominal throughput of up to 20,000 transactions per second~\cite{Gorenflo2020}, yet even an unreasonably fast transaction life cycle of \SI{50}{\milli\second} from execution to commitment would result in a maximum of 20 updates per second to the same key, or once every ten blocks with a block size of 100 transactions. Worse yet, transactions are not only invalidated if their RW set overlaps completely, but also if there is a \emph{single} key overlap with a previous transaction. This means that workloads with hot keys can easily reduce effective throughput by several orders of magnitude.

While early abort schemes can discard invalid transactions before they become part of a block, they cannot break the theorem. Assuming they result in blocks without invalid transactions, they can only fill up the slots in a new block with transactions using different key spaces. Thus, they skew the processed transaction distribution.
%and do not reflect the actual demand any more. 
Furthermore, aborted transactions need to be re-submitted and re-executed, flooding the network with even more attempts to modify hot keys. Eventually, endorsers will be completely saturated by clients repeatedly trying to get their invalid transactions re-executed.
%If no other mechanisms are put into place, this will lead to a complete blockage of endorsers by clients trying to get their invalid transactions re-executed in a short amount of time.
    
\section{The XOX hybrid model}
\label{sec:xox}

To deal with the drawbacks of both the OX and XO patterns, we now present the \emph{execute-order-execute~(XOX)} pattern which adds a secondary post-order execution step to execute
the patch-up code added to smart contracts.
XOX minimizes transaction conflicts while preserving concurrent block processing and without the introduction of any centralized elements. 
In this section, we first describe the necessary changes to the endorsers' pre-order execution step to allow the inclusion of external oracles in the post-order execution step. Then, we describe changes to the critical transaction flow path on the peers after they receive blocks from the ordering service. The details of the crucial steps we introduce are described in sections~\ref{sec:dep} and~\ref{sec:x}. Notably, our changes do not affect the ordering service, preserving Fabric's modular structure.

\subsection{Pre-order endorser execution}
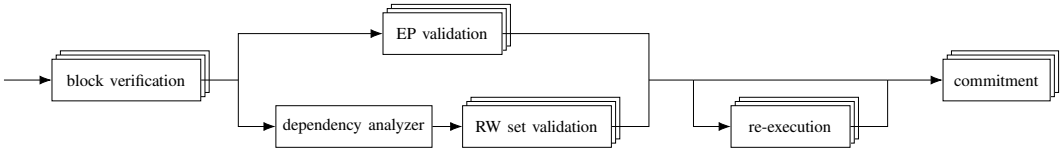
\begin{figure*}
    \centering
    \resizebox{0.8\textwidth}{!}{
    	\begin{tikzpicture}[
	    >=serif cm,font = \footnotesize,
	    node distance=0.3cm and 0.5cm,
	    map/.style={draw, rectangle, align=left, minimum height = 7mm   },
	    arrow/.style={-{Latex[length=2mm]}},
	]
	\node[](ord){};
	
	\node[map,right=0.8cm of ord, fill=white, minimum width=25mm](v2){block verification};
	\begin{scope}[on background layer]
        \node[map,above right=2mm of v2.south west, fill=white, minimum width=25mm](v1){};
	    \node[map,above right=1mm of v2.south west, fill=white, minimum width=25mm](v3){};
    \end{scope}

    \draw[arrow](ord)--(v2);

	\node[right=of v2](c1){};
	
	\draw[-](v2) --(c1.center);
	
	\node[map,above right=0.3cm and 2.3cm of c1,fill=white, minimum width=20mm](ep2){EP validation};
	
	\draw[arrow](c1.center) |-(ep2);
	
	\begin{scope}[on background layer]
    	\node[map,above right=2mm of ep2.south west,fill=white, minimum width=20mm](ep1){};
    	\node[map, above right=1mm of ep2.south west,fill=white, minimum width=20mm](ep3){};
    \end{scope}

	\node[map,below right=of c1](dep){dependency analyzer};
	
	\draw[arrow](c1.center) |-(dep);
	
	\node[map,right=of dep,fill=white, minimum width=25mm](rw){RW set validation};
	
	\draw[arrow](dep) --(rw);
	
	\begin{scope}[on background layer]
    	\node[map,above right=2mm of rw.south west,fill=white,, minimum width=25mm]{};
    	\node[map, above right=1mm of rw.south west,fill=white,, minimum width=25mm]{};
    \end{scope}
    
    \node[above right=of rw](c2){};
    
    \draw[-](ep2)-|(c2.center);
    \draw[-](rw)-|(c2.center);
    \node[right=of c2](c3){};
    
    \draw[-](c2.center)--(c3.center);
    
    \node[map,below right=of c3,fill=white, minimum width=20mm](x){re-execution};
    
    \begin{scope}[on background layer]
    	\node[map,above right=2mm of x.south west,fill=white,, minimum width=20mm]{};
    	\node[map, above right=1mm of x.south west,fill=white,, minimum width=20mm]{};
    \end{scope}
    
    \node[above right=of x](c4){};
    
    \node[map,right=0.8cm of c4, fill=white, minimum width=18mm](com){commitment};
    
       \begin{scope}[on background layer]
    	\node[map,above right=2mm of com.south west,fill=white,, minimum width=18mm]{};
    	\node[map, above right=1mm of com.south west,fill=white,, minimum width=18mm]{};
    \end{scope}
    
    \draw[arrow](c3.center)|-(x);
    \draw[-](c3.center)--(c4.center);
    \draw[-](x)-|(c4.center);
    \draw[arrow](c4.center)--(com);

	\end{tikzpicture}
    }
    \caption{The modified XOX Fabric validation and commitment pipeline. Stacks and branched paths show parallel execution.}
    \label{fig:pipeline}
\end{figure*}
The pre-order execution step leverages concurrent transaction execution and uses general purpose programming languages like \emph{Go}. Depending on the endorsement policy, clients request multiple endorsers to execute their transaction and the returned execution results must be identical. This makes a deterministic execution environment unnecessary because deviations are discarded and a unanimous result from all endorsers becomes ground truth for the whole network. Notably, this also allows external oracles like weather or financial data. If these oracle data lead to non-deterministic RW sets, the client will not receive identical endorser responses and the transaction will never reach the Fabric network.

External oracles are a powerful tool. If they are supported by the pre-order execution step, they must also be supported by the post-order execution step.
To achieve this, we must make the oracle deterministic. We leverage the same mechanism that ensures deterministic transaction results for  pre-order execution: We extend the transaction response by an additional \emph{oracle set}. Any external data are recorded in the form of key-value pairs and are added to the response to the client. Now, if the oracle sets for the same transaction executed by different endorsers differ, the client has to discard the transaction. Otherwise, the external data effectively becomes part of the deterministic world state so that it can be used without by the post-order execution step. Analogous to existing calls to \texttt{GetState} and \texttt{PutState} that record the read and write set key-value pairs, respectively, we add a new call \texttt{PutOracle} to the chaincode API.

\subsection{Critical transaction flow path}
Our previous work on FastFabric~\cite{Gorenflo2020} showed how to improve performance by pipelining the syntactic block verification and endorsement policy validation (EP validation) so that it can be done for multiple blocks at the same time. However, the RW set validation to check for invalid state transitions and the final commitment had to be done sequentially in a single thread. While the the XOX model is an orthogonal optimization, its second execution step needs to be placed between RW set validation and commitment. Since this step is relatively slow, we must expand our concurrency efforts to pipelining RW set validation, post-order executions, and commitment. Two vital pieces for this effort, a transaction dependency analyzer and the executions step itself, are described in later sections in detail, so we will only give a brief overview here. This allows us to concentrate on the pipeline integration in this section.

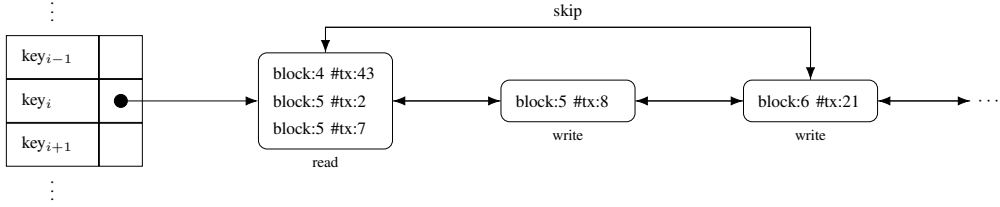
\begin{figure*}
    \centering
    \resizebox{0.75\textwidth}{!}{
    	\begin{tikzpicture}[
	    >=serif cm,font = \footnotesize,
	    map/.style={draw, rectangle, align=left, node distance = 7mm and 0.1cm, minimum height = 7mm, minimum width=1.5cm, text width=1cm, node distance=0mm and 0mm
	    },
	    list/.style={node distance = 0mm and 2cm, text width = 1.7cm, minimum width=1.5cm},
	    arrow/.style={-{Latex[length=2mm]}},
	]
	\node[map](key0){key$_{i-1}$};
	\node[above=1mm of key0]{\vdots};
	\node[map,below=of key0](key1){key$_i$} ;
	
	\node[map,right=of key1,minimum width=7mm, text width = 0cm](anchor){};
	\filldraw [black] (anchor) circle (3pt);
	\node[map,above=of anchor,minimum width=7mm, text width = 0cm]{};
	\node[map,below=of anchor,minimum width=7mm, text width = 0cm]{};
	\node[map,below=of key1](key2){key$_{i+1}$} ;
	\node[below=-1mm of key2]{\vdots};
	\node[list,right=of anchor](l2){block:5 \#tx:2};
	\node[list,above=of l2](l1){block:4 \#tx:43};
	\node[list,below=of l2](l3){block:5 \#tx:7};
	
	\node[rectangle,rounded corners, draw, fit=(l1)(l2)(l3)](r){};
	\node[below right=0mm of r.south, anchor=north, font=\scriptsize]{read};
	\node[list,right=of l2](l4){block:5 \#tx:8};
	\node[rectangle,rounded corners, draw, fit=(l4)](w1){};
	\node[below right=0mm of w1.south, anchor=north, font=\scriptsize]{write};
	
	\node[list,right=of l4](l5){block:6 \#tx:21};
	\node[rectangle,rounded corners, draw, fit=(l5)](w2){};
	\node[below right=0mm of w2.south, anchor=north, font=\scriptsize]{write};
	
	\draw[arrow] (anchor.center) -- (r);
	\draw[arrow] (r) -- (w1);
	\draw[arrow] (w1) -- (r);
	\draw[arrow] (w1) -- (w2);
	\draw[arrow] (w2) -- (w1);
	\draw[arrow] (r) -- ($(r.center) + (0,1.2)$) -- ($(w2.center) + (0,1.2)$)node [midway, above]{skip} -- (w2);
	\draw[arrow] (w2) -- ($(w2.center) + (0,1.2)$) -- ($(r.center) + (0,1.2)$)-- (r);
	
	\node[right=1.5cm of w2](etc){$\cdots$};
	\draw[arrow] (w2) -- (etc);
	\draw[arrow] (etc) -- (w2);
	
	\end{tikzpicture}
    }
    \caption{Dependency analyzer data structure: Example of a state database key mapped to a doubly-linked skiplist of dependent transactions. }
    \label{fig:dp}
\end{figure*}

\subsubsection{Dependency analyzer}
For concurrent transaction processing, we rely on the ability to isolate transactions from each other. However, the sequential order of transactions in a block matters when their RW sets are validated and they are committed. A dependency exists when two transactions overlap in some keys of their RW sets (read-only transactions do not even enter the orderer). In that case, we cannot process them independently. Therefore, we need to keep track of dependencies between transactions so we know which subsets of transactions can be processed concurrently. 

\subsubsection{Execution step}
Transactions for which the dependency analyzer has found a dependency on an earlier transaction would be invalidated during Fabric's RW set validation. We introduce a step which re-executes transaction with such an RW set conflict based on the most up-to-date world state. It can resolve conflicts due to a lack of knowledge of concurrent transactions during  pre-order execution. However, it still invalidates transactions that attempt something the smart contract does not allow, such as creating a negative account balance.

In FastFabric, peers receive blocks as fast as the ordering service can deliver them. If the syntactic verification of a block fails, the whole block is discarded. Thus, it is reasonable to keep this as a first step in the pipeline. Next up is the EP validation step. Each transaction can be validated in parallel because the validations are independent of each other. The next step is the intertwined RW set validation and commitment: Each transaction is validated, and, if successful, added to an update batch that is subsequently committed to the world state.

XOX Fabric separates RW set validation from the commitment decision. Therefore, this step is no longer dependent on the result of the EP validation and can be done in parallel. However, in order to validate transactions concurrently, we need to know their dependencies, so the dependency analyzer goes first and releases transactions to the RW set validation as their dependencies are resolved.

Subsequently, the results from the EP validation and RW set validation are collected, and if they are marked as valid, they can be committed concurrently. If a RW set conflict arises, they need to be sent to the new execution step to be re-executed based on the current world state. Finally, successfully re-executed transactions are committed and all others are discarded. 

Our design allows dependency analysis to work in parallel to endorsement policy validation and transactions can proceed as soon as all previous dependencies are known. Specifically, independent sets of transactions can pass through RW set validation, post-order execution, and commitment steps concurrently. The modified pipeline is shown in Fig.~\ref{fig:pipeline}.

\section{Dependency analyzer}
\label{sec:dep}
%It is unnecessary to force global transaction serialization. As previously discussed, sets of independent transactions can be processed concurrently. However, to obtain this dependency information, we need to introduce a new mechanism into the critical path of the peers.
We now discuss the details of the dependency analyzer. Note that the only way for a transaction to have a dependency on another is an overlap in its RW set with a previous transaction. More precisely, one of the conflicting transactions must write to the overlapping key. Reads do not change the version nor the value of a key, so they do not impede each other. However, we must consider a write a blocking operation for that key. If transaction $a$ with a write is ordered before transaction $b$ with a read from the same key, then this must always happen in this order lest we lose deterministic behaviour of the peer because of the changing key value. The reverse case of read-before-write has the same constraints. In the write-write case,  neither transaction actually relies on the version or the value of that key. Nevertheless, they must remain in the same order, otherwise transaction $a$'s value might win out, even though transaction $b$ should overwrite it.

To detect such conflicts, we keep track of read and write accesses to all keys across transactions. For each key, we create a doubly-linked skip list that acts as a dependency queue, recording all transactions that need to access it. Entries in this queue are sorted by the blockchain transaction order. As described before, consecutive reads of the same key do not affect each other and can be collapsed into a single node in the linked list so they will be freed together. For faster traversal during insertion, nodes can skip to the start of the next block in the list. This data structure is illustrated in Fig.~\ref{fig:dp}. After the analysis of a transaction is complete, it will not continue to the next step in the pipeline until all previous transactions have also been analyzed, lest an existing dependency might be missed.

Dependencies may change in two situations: when new transactions are added or existing transactions have completed the commitment pipeline. In either case, we update the dependency lists accordingly and check the first node of lists that have been changed. If any of these transactions have no dependency in any key anymore, they are released into the validation pipeline. However, we can only remove a transaction from the dependency lists once it is either committed or discarded, lest dependent transactions get freed up prematurely.

\section{Post-order execution step}
\label{sec:x}
The post-order execution step executes additional patch-up code added to a smart contract. We discuss it in more detail in this section. 

When the RW validation finds a conflict between a transaction's RW set and the world state, that transaction will be re-executed and possibly salvaged using the patch-up code. However, the post-order execution stage needs to adhere to some constraints. First, the new RW set must be a subset of the original RW set so the dependency analyzer can reason properly. Without this restriction, new dependencies could emerge and transactions scheduled for parallel processing would now create an invalid world state. Second, the blockchain network also needs consistency among peers. Therefore, the post-order execution must be deterministic so there is no need for further consensus between peers. Lastly, this new execution step is part of the critical path and thus should be as fast as possible.

For easier adoption of smart contracts from other blockchain systems, we use a modified version of Ethereum's EVM~\cite{Founder2017} as the re-execution engine for patch-up code\footnote{We note that forays have been made to build WebAssembly based execution engines~\cite{EWASM}, which would allow for a variety of programming languages to build smart contracts for the post-order execution step.}. Patch-up code take a transaction's read set and oracle set as input. The read set is used to get the current key values from the latest version of the world state. Based on this and the oracle set, the smart contract then performs the necessary computations to generate a new write set. If the transaction is not allowed by the logic of the smart contract based on the updated values, it is discarded. Finally, in case of success, it generates an updated RW set, which is then compared to the old one. If all the keys are a subset of the old RW set, the result is valid and can be committed to the world state and blockchain.

For example, suppose client $A$ wants to add 70 digital coins to an account with a current balance of 20 coins. Simultaneously, client $B$ wants to add 50 coins to the same account. They both have to read the key of the account, update its value, and write the new value back, so the account's key is in both transactions' RW set. Even if both clients are honest, only the transaction which is ordered earlier would be committed. Without loss of generality, assume that $A$'s transaction updates the balance to 90 coins because it won the race. In XOX Fabric, $B$'s transaction would wait for $A$ to finish due to its dependency and then would find a key version conflict in the RW validation step. Therefore, it is sent to the post-order execution step. In the step, $B$'s patch-up code can read the updated value from the database and add its own value for a total of 140 coins, which is recorded in its write set. After successful execution, the RW set comparison is performed and the new total will be committed. Thus, the re-execution of the patch-up code salvages conflicting transactions.

However, if we start with an account balance of 100 coins and $A$ tries to subtract 50 coins and $B$ tries to subtract 60 coins, we get a different result. Again, $B$'s transaction would be sent to be re-executed. But this time, it's patch-up code tries to subtract 60 coins from the updated 50 coins and the smart contract does not allow a negative balance. Therefore, $B$'s transaction will be discarded, even though it was re-executed based on the current world state.

Thus, our hybrid XOX approach can correct transactions which would have been discarded because they were executed based on a stale world state. However, transactions that do not satisfy the smart contract logic are still rejected.
%which lead to a world state which is undesired by the smart contract logic keep being invalidated.

Lastly, if we do not put any restrictions on the execution, we risk expensive computations, low throughput, and even non-terminating smart contracts. Ethereum deals with this by introducing \emph{gas}. If a smart contract runs out of gas, it is aborted and the transaction is discarded. As of yet, Fabric does not include such a concept.

As a solution, we introduce \emph{virtual gas} as a tuning parameter for system performance. Instead of originating from a bid by the client that proposes the transaction, it can be set by a system administrator. If the post-order step runs out of gas for a transaction, it becomes immediately invalidated, but in case of success the fee is never actually paid. A larger value allows for more complex computation at the cost of throughput. While the gas parameter should generally be as small as possible, large values could make sense for workloads with very infrequent transaction conflicts and high importance of conflict resolution.

\section{Experiments}

\label{sec:exp}
We now evaluate the peformance of XOX Fabric. We used 11 local servers connected by a 1~Gbit/s switch. Each is equipped with two Intel\textsuperscript{\tiny\textregistered}~Xeon\textsuperscript{\tiny\textregistered}~CPU~E5-2620~v2 processors at 2.10~GHz, for a total of 24 hardware threads and 64~GB of RAM. 
We compare three systems with different capabilities. Fabric 1.4 is the baseline. Next, FastFabric~\cite{Gorenflo2020} adds efficient data structures, improved parallelization, and decoupled endorsement and storage servers. Finally, our implementation of an XOX model based on FastFabric adds transaction dependency analysis, concurrent key version validation, and transaction re-execution.

For comparable results, we match the network setup of all three systems as closely as possible. We use a single orderer in \emph{solo} mode, ensuring that throughput is bound by the peer performance. A single anchor peer receives blocks from the orderer and broadcasts them to four endorsing peers. In the case of FastFabric and XOX, the broadcast includes the complete transaction validation metadata so endorsers can skip their own validation steps. FastFabric and XOX run an additional persistent storage server because in these cases the peers store their internal state database in-memory. The remaining four servers are used as clients\footnote{We do not use Caliper because it is not sufficiently high-performance to fully load our system.}. Spawning a total of 200 concurrent threads, they use the Fabric node.js SDK to send transaction proposals to the endorsing peers and consecutively submit them to the orderer. Each block created by the orderer contains 100 transactions. 

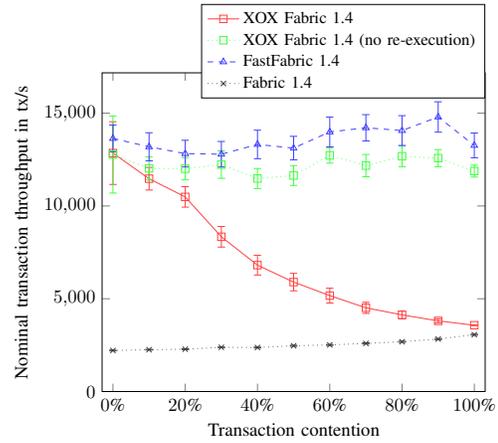
\begin{figure}
    \centering
    \begin{tikzpicture}[scale = 0.7]
        \begin{axis}[
        scaled ticks=false, 
        tick label style={/pgf/number format/fixed},
        legend style={at={(1,0.92)}, anchor=south east, font=\small},
        legend cell align=left,
	    enlarge x limits=0.03,
        ymin=0,
        width= \linewidth,
        bar width=0.1cm,
        ylabel= Nominal transaction throughput in tx/s,
        xlabel = {Transaction contention},
        xticklabel={\pgfmathprintnumber\tick\%},
        ]       
        
        \addplot[red!75, mark=square, error bars/.cd, y dir = both,y explicit] table[x=x, y=xox, y error=xox_err]{tikz/nominal_data.dat};
        \addplot[green!75, mark options={solid}, dotted, mark=square, error bars/.cd, y dir = both,y explicit,error bar style={solid}] table[x=x, y=xox2, y error=xox2_err]{tikz/nominal_data.dat};
        \addplot[dashed, blue!75, mark options={solid} ,mark=triangle,error bars/.cd, y dir = both,y explicit,error bar style={solid}] table[x=x, y=fastfabric,y error=ff_err]{tikz/nominal_data.dat};
        
        \addplot[dotted, mark options={solid},black!75,mark=x, error bars/.cd, y dir = both,y explicit,error bar style={solid}, error mark options={solid,rotate=90, mark size=2}] table[x=x,y=fabric,y error=fabric_err]{tikz/nominal_data.dat};
        \legend{XOX Fabric 1.4, XOX Fabric 1.4 (no re-execution),  FastFabric 1.4,Fabric 1.4}
        \end{axis}
        \end{tikzpicture}
    \caption{Impact of transaction conflicts on nominal throughput, counting both valid and invalid transactions.}
    \label{fig:nom_tp}
\end{figure}

All experiments run the same chaincode: A money transfer  from  one  account  to  another is simulated, reading from and writing to two keys in the state database, e.g. deducting 1 coin from \emph{account0} and adding 1 coin to \emph{account1}. We use the default endorsement  policy  of  accepting  a  single endorser signature. XOX's second execution phase integrates a Python virtual stack machine (VM) implemented in Go~\cite{Gpython}. We added a parameter to the VM to stop the stack machine after executing a certain amount of operations, emulating a \emph{gas} equivalent. We load a Python implementation of the Go chaincode into the VM and extract the call parameters from the transaction so that the logic between pre-order and post-order execution remains the same. Therefore, the only semantic difference between XO and OX is that OX operates on up-to-date state.

For each tested system, clients generate a randomized load with a specific contention factor by flipping a loaded coin for each transaction. Depending on the outcome, they either choose a previously unused account pair or the pair \emph{account0--account1} to create a RW set conflict. We scale the transaction contention factor from 0\% to 100\% in 10\% steps and run the experiment for each of the three systems. Every time, clients generate a total of 1.5 Million transaction. In the following, we will discuss XOX's throughput improvements under contention over both FastFabric and Fabric 1.4, and its overhead compared to FastFabric.

\subsection{Throughput}

We start by examining the nominal throughput of each system in Fig.~\ref{fig:nom_tp}. We measured the throughput of all transactions regardless of their validity. The effectively single-threaded validation/commitment pipeline of Fabric 1.4 creates results with little variance over time. The throughput increases slightly from about 2200 tx/s to 3000 tx/s the higher the transaction contention becomes, because Fabric discards invalid transactions, so their changes are not committed to the world state database. FastFabric follows the same trend, going from 13600 tx/s up to 14700 tx/s, although the relative throughput increase is not as pronounced because the database commit is cheaper, and there is  higher variance due to many parallel threads vying for resources at times.

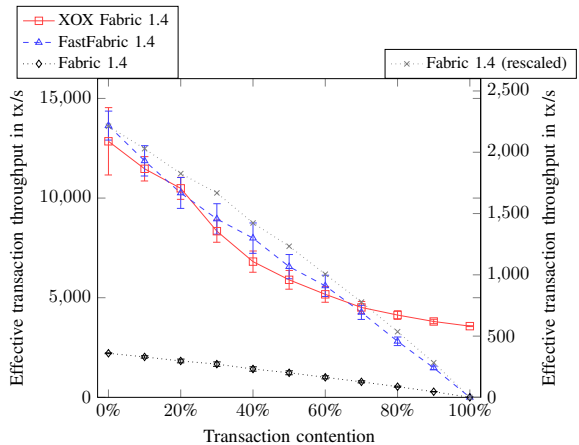
\begin{figure}
    \centering
    \begin{tikzpicture}[scale = 0.7]
        \begin{axis}[
        scaled ticks=false, 
        tick label style={/pgf/number format/fixed},
        legend style={at={(0,1)}, anchor=south, font=\small},
        legend cell align=left,
	    enlarge x limits=0.03,
        ymin=0,
        width= \linewidth,
        bar width=0.1cm,
        ylabel= Effective transaction throughput in tx/s,
        xlabel = {Transaction contention},
        xticklabel={\pgfmathprintnumber\tick\%},
        ]       
        
        \addplot[red!75, mark=square, error bars/.cd, y dir = both,y explicit] table[x=x, y=xox, y error=xox_err]{tikz/effective_data.dat};
        \addplot[dashed, blue!75, mark options={solid} ,mark=triangle,error bars/.cd, y dir = both,y explicit,error bar style={solid}] table[x=x, y=fastfabric,y error=ff_err]{tikz/effective_data.dat};
        \addplot[dotted, mark options={solid},black,mark=diamond, error bars/.cd, y dir = both,y explicit,error bar style={solid}, error mark options={solid,rotate=90, mark size=2}] table[x=x,y=fabric,y error=fabric_err]{tikz/effective_data.dat};
        \legend{XOX Fabric 1.4, FastFabric 1.4, Fabric 1.4}
        \end{axis}
        \begin{axis}[
        scaled ticks=false, 
        tick label style={/pgf/number format/fixed},
        legend style={at={(1.0,1)}, anchor=south, font=\small},
        legend cell align=left,
	    enlarge x limits=0.03,
        ymin=0,
        ymax=2600,
        width= \linewidth,
        bar width=0.1cm,
        ylabel= Effective transaction throughput in tx/s,
        axis y line*=right,
        axis x line=none
        ] 
        \addplot[dotted, mark options={solid},black!50,mark=x, error bars/.cd, y dir = both,y explicit,error bar style={solid}, error mark options={solid,rotate=90, mark size=2}] table[x=x,y=fabric]{tikz/effective_data.dat};
        \legend{Fabric 1.4 (rescaled)}
        \end{axis}
        \end{tikzpicture}
    \caption{Impact of transaction conflicts on effective throughput, counting only valid transactions. Fabric 1.4 scaled up for slope comparison (right y-axis).}
    \label{fig:eff_tp}
\end{figure}

We ran the experiments for XOX in two configurations to  understand the effects of different parts of our implementation on the overall throughput. First, we only included changes to the validation pipeline and the addition of the dependency analyzer but disabled transaction re-execution. Subsequently, we ran it again with all features enabled. The first configuration shows roughly the same behaviour as FastFabric, albeit with a small hit to overall throughput, ranging from 12000 tx/s up to 12700 tx/s. For higher contention ratios, the fully-featured configuration's throughput drops from 12800 tx/s to about 3600 tx/s, a third of its initial value. However, this is expected as more transactions need to be re-executed sequentially. Importantly, even under full contention, XOX performs better than Fabric 1.4.

Note that the nominal throughout is meaningless if blocks contain mostly invalid transactions. Therefore, we now discuss the effective throughput. In Fig.~\ref{fig:eff_tp}, we have eliminated all invalid transactions from the throughput measurements. Naturally, this means there is no change for the full XOX implementation, because it already produces only valid transactions. For better comparison of the three systems under contention, we normalized the projections of their plots. FastFabric and XOX follow the left y-axis while Fabric follows the right one. For up to medium transaction contention, all systems roughly follow the same slope. However, while both FastFabric and Fabric tend towards 0 tx/s in the limit of 100\% contention, XOX still reaches a throughput of about 3600 tx/s. At this point, all transaction in a submitted block have to be re-executed. This means, starting at 70\% contention, XOX surpasses all other systems in terms of effective throughput while maintaining comparable throughput before that threshold.

\begin{figure}
    \centering
    \hspace*{-0.8cm}
    \begin{tikzpicture}[scale = 0.7, xshift=10cm]
        \begin{axis}[
        scaled ticks=false,
        legend style={at={(1,1)}, anchor=north east, font=\small},
        legend cell align=left,
	    enlarge x limits=0.03,
        ymin=0,
        ymax=100,
        width= \linewidth,
        bar width=0.1cm,
        ylabel= Relative nominal load,
        xlabel = {Transaction contention},
        xticklabel={\pgfmathprintnumber\tick\%},
        yticklabel={\pgfmathprintnumber\tick\%}
        ]       
        
        \addplot[red!75, mark=x, error bars/.cd, y dir = both,y explicit] table[x=x, y=ex_ov, y error=ex_err]{tikz/nominal_data.dat};
        \addplot[dashed, blue!75, mark options={solid} ,mark=triangle,error bars/.cd, y dir = both,y explicit,error bar style={solid}] table[x=x, y=pip_ov, y error=pip_err]{tikz/nominal_data.dat};
        \legend{ post-order execution,dependency analyzer \& pipeline, test};
        \end{axis}
        \end{tikzpicture}
    \caption{Relative load overhead of separate XOX parts over FastFabric.}
    \label{fig:nom_load}
\end{figure}
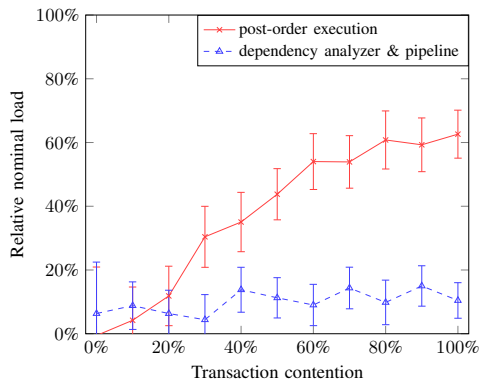

Even though it might seem like a corner case, this is a significant improvement. All experiments were run with a synthetic static workload where the level of contention stayed constant. However, in a real world scenario, users have two options when their transaction fails. They can abandon the transaction, or, more likely, submit the same transaction again. In a system with some amount of contention, conflicting transactions can accumulate over time by users resubmitting them repeatedly. This results in an unintended denial of service attack. In contrast, XOX guarantees liveness in every scenario.

\subsection{Overhead}

We now explore the overhead of XOX compared to FastFabric's nominal performance in Fig.~\ref{fig:nom_load}. We isolate the overhead introduced by adding the dependency analyzer and modifying the validation pipeline so that it can handle single transactions instead of complete blocks, as well as the overhead of the transaction re-execution by the python VM.

The blue dashed line shows that the dependency analyzer overhead is almost constant regardless of contention level. By minimizing spots in the validation/commitment pipeline that require a sequential processing order of transactions, we achieve an overhead of less than 15\% even when the dependency analyzer only releases one transaction at a time.

In contrast, the overhead of the re-execution step is more noticeable. For high contention, this step generates over 60\% of additional load. Yet, this also means that replacing the highly inefficient Python VM used in our proof of concept with a faster deterministic execution environment could dramatically increase XOX's throughput for high-contention workloads. This would push the threshold when XOX beats the other systems to lower fractions of conflicting transactions. Furthermore, the contention load used for these experiments presents the absolute worst case, where every conflicting transaction is touching the same state keys, resulting in a fully sequential re-execution of all transactions. However, if instead of a single account pair \emph{account0--account1} used for contentious transactions there was a second pair \emph{account2--account3}, the OX step would run in two concurrent threads instead of one. Even with this simple relaxation, the overhead would roughly be cut in half.

\section{Conclusion and Future Work}
In this work, we propose a novel hybrid execution model for Hyperledger Fabric consisting of a pre-order and a post-order execution step. This allows a trade-off between parallel transaction execution and minimal invalidation due to conflicting results. In particular, our solution can deal with highly skewed workloads where most transactions use only a small set of hot keys. Contrary to other post-order execution models, we support the use of external oracles in our secondary execution step. We show that the throughput of our implementation scales comparably to Fabric and FastFabric for low contention workloads, and surpasses them when transaction conflicts increase in frequency.

Now that all parts of the validation and commitment pipeline are decoupled and highly scalable, it remains to be seen in future work if the pipeline steps can be scaled across multiple servers to improve the throughput further.

 \bibliography{Mendeley}

\begin{thebibliography}{10}

\bibitem{Amiri2019b}
Mohammad~Javad Amiri, Divyakant Agrawal, and Amr~El Abbadi.
\newblock {CAPER: a cross-application permissioned blockchain}.
\newblock {\em Proceedings of the VLDB Endowment}, 12(11):1385--1398, 2019.

\bibitem{Amiri2019}
Mohammad~Javad Amiri, Divyakant Agrawal, and Amr El~Abbadi.
\newblock {ParBlockchain: Leveraging Transaction Parallelism in Permissioned
  Blockchain Systems}.
\newblock {\em Proceedings - International Conference on Distributed Computing
  Systems}, 2019-July:1337--1347, 7 2019.

\bibitem{Androulaki2018a}
Elli Androulaki, Artem Barger, Vita Bortnikov, Christian Cachin, Konstantinos
  Christidis, Angelo De~Caro, David Enyeart, Christopher Ferris, Gennady
  Laventman, Yacov Manevich, Srinivasan Muralidharan, Chet Murthy, Binh Nguyen,
  Manish Sethi, Gari Singh, Keith Smith, Alessandro Sorniotti, Chrysoula
  Stathakopoulou, Marko Vukoli{\'{c}}, Sharon~Weed Cocco, and Jason Yellick.
\newblock {Hyperledger Fabric: A Distributed Operating System for Permissioned
  Blockchains}.
\newblock {\em Proceedings of the Thirteenth EuroSys Conference on - EuroSys
  '18}, pages 1--15, 2018.

\bibitem{Chen2018}
Si~Chen, Jinyu Zhang, Rui Shi, Jiaqi Yan, and Qing Ke.
\newblock {A comparative testing on performance of blockchain and relational
  database: Foundation for applying smart technology into current business
  systems}.
\newblock In {\em International Conference on Distributed, Ambient, and
  Pervasive Interactions}, pages 21--34. Springer Verlag, 2018.

\bibitem{Dinh2017d}
Tien Tuan~Anh Dinh, Ji~Wang, Gang Chen, Rui Liu, Beng~Chin Ooi, and Kian-Lee
  Tan.
\newblock {BLOCKBENCH: A Framework for Analyzing Private Blockchains}.
\newblock {\em Proceedings of the 2017 ACM International Conference on
  Management of Data - SIGMOD '17}, pages 1085--1100, 2017.

\bibitem{Escobar2019}
Ian~Aragon Escobar, Eduardo~E.P. Alchieri, Fernando~Luís Dotti, and Fernando
  Pedone.
\newblock {Boosting concurrency in Parallel State Machine Replication}.
\newblock In {\em Proceedings of the 20th International Middleware Conference},
  pages 228--240. Association for Computing Machinery (ACM), 2019.

\bibitem{EWASM}
{Ethereum Community}.
\newblock {EWASM}, 2018.

\bibitem{Gpython}
{go-python}.
\newblock {Python 3.4 interpreter implementation for Golang}.

\bibitem{Gorenflo2020}
Christian Gorenflo, Stephen Lee, Lukasz Golab, and Srinivasan Keshav.
\newblock {FastFabric: Scaling Hyperledger Fabric to 20,000 Transactions per
  Second}.
\newblock In {\em International Journal of Network Management}. John Wiley and
  Sons Ltd, 2 2020.

\bibitem{Nasirifard2019}
Pezhman Nasirifard, Ruben Mayer, and Hans-Arno Jacobsen.
\newblock {FabricCRDT: A Conflict-Free Replicated Datatypes Approach to
  Permissioned Blockchains}.
\newblock In {\em Proceedings of the 20th International Middleware Conference},
  pages 110--122. Association for Computing Machinery (ACM), 2019.

\bibitem{Schneider1990}
Fred~B. Schneider.
\newblock {Implementing Fault-Tolerant Services Using the State Machine
  Approach: A Tutorial}.
\newblock {\em ACM Computing Surveys (CSUR)}, 22(4):299--319, 1990.

\bibitem{Sharma2019}
Ankur Sharma, Felix~Martin Schuhknecht, Divyakant Agrawal, and Jens Dittrich.
\newblock {Blurring the Lines between Blockchains and Database Systems}.
\newblock In {\em Proceedings of the 2019 International Conference on
  Management of Data}, pages 105--122, 2019.

\bibitem{Sorniotti}
Alessandro Sorniotti, Angelo De~Caro, Baohua Yang, Binh Nguyen, Manish Sethi,
  Vukolic Marko, Sheehan Anderson, Srinivasan Muralidharan, and Parth Thakkar.
\newblock {Fabric Proposal: Enhanced Concurrency Control}, 2017.

\bibitem{Founder2017}
Gavin Wood.
\newblock {Ethereum: a Secure Decentralised Generalised Transaction Ledger},
  2014.

\bibitem{Zhang2019}
Shenbin Zhang, Ence Zhou, Bingfeng Pi, Jun Sun, Kazuhiro Yamashita, and
  Yoshihide Nomura.
\newblock {A Solution for the Risk of Non-deterministic Transactions in
  Hyperledger Fabric}.
\newblock {\em IEEE International Conference on Blockchain and Cryptocurrency
  (ICBC)}, pages 253--261, 2019.

\end{thebibliography}
 \bibliographystyle{plain}
%\printbibliography

\end{document}